\documentclass{article}

\usepackage{arxiv}
\usepackage{amsmath}
\usepackage{amsthm}
\usepackage[utf8]{inputenc} % allow utf-8 input
\usepackage[T1]{fontenc}    % use 8-bit T1 fonts
\usepackage{hyperref}       % hyperlinks
\usepackage{url}            % simple URL typesetting
\usepackage{booktabs}       % professional-quality tables
\usepackage{amsfonts}       % blackboard math symbols
\usepackage{nicefrac}       % compact symbols for 1/2, etc.
\usepackage{microtype}      % microtypography
\usepackage{lipsum}
\usepackage{natbib}

\usepackage[svgnames]{xcolor}
\usepackage{listings}
\newtheorem{theorem}{Theorem}[section]

\newtheorem{lemma}[theorem]{Lemma}

\title{The equivalence of the Delta method and the cluster-robust variance estimator for the analysis of clustered randomized experiments}

\author{
    Alex Deng\\
    Airbnb Inc.\\
    alex.deng@airbnb.com
  \and
    Jiannan Lu\\
    Microsoft Corporation\\
    jiannl@microsoft.com
    \and
    Wen Qin\\
    Microsoft Corporation\\
    weqin@microsoft.com
}

\begin{document}
\maketitle

\begin{abstract}
It often happens that the same problem presents itself to different communities and the solutions proposed or adopted by those communities are different. We take the case of the variance estimation of the population average treatment effect in cluster-randomized experiments. The econometrics literature promotes the cluster-robust variance estimator \citep{athey2017econometrics}, which can be dated back to the study of linear regression with clustered residuals \citep{liang1986longitudinal}. The A/B testing or online experimentation literature promotes the delta method \citep{kohavi2010online,Deng2017,Deng2018}, which tackles the variance estimation of the ATE estimator directly using large sample theory. The two methods are seemly different as the former begins with a regression setting at the individual unit level and the latter is semi-parametric with only i.i.d. assumptions on the clusters. Both methods are widely used in practice. It begs the question for their connection and comparison. In this paper we prove they are equivalent and in the canonical implementation they should give exactly the same result. 
\end{abstract}

% keywords can be removed

\section{Introduction}
Consider the problem of analyzing the population average treatment effect in a cluster-randomized experiment. Let $(Y_{i}, W_i), i=1,\dots,N$ be the observation and treatment assignment pairs for $N$ subjects. An unbiased estimator for the population average treatment effect is simply the difference of the averages of the two variants:
$$
\Delta = \frac{\sum_{W_i = 1} Y_i}{N_T}-\frac{\sum_{W_i = 0} Y_i}{N_C}\ ,
$$
where $N_T = \sum_i W_i$ and $N_C = \sum_i (1-W_i)$ are the treatment unit count and control unit count. In particular, we are interested in estimating its variance. Since treatment and control groups are independent, 
\begin{equation}\label{eq:var}
\mathrm{Var}(\Delta) = \mathrm{Var}\left(\frac{\sum_{W_i = 1} Y_i}{N_T} \right) + \mathrm{Var}\left(\frac{\sum_{W_i = 0} Y_i}{N_C} \right)\ .
\end{equation}

What makes the cluster-randomized experiments special is that the treatment assignment $W_i$ are not independent to each other because the randomization is operated on the cluster level, not individual unit level \citep{klar2001current}. Let $G$ be the number of clusters, and let $G_{i g} \in\{0,1\}, i=1, \ldots, N, g=1, \ldots, G$ denote the binary indicator that unit $i$ belongs to cluster $g$. In this paper we assume that each unit only belongs to one cluster, i.e., $\sum_{g=1}^G G_{ig}=1$ for all $i=1, \ldots, N.$ Furthermore, we let $N_{g}=\sum_{i=1}^{N} G_{i g}$ denote the number of units in cluster $g$. For any given cluster $g$, $W_i$ are either all $1$ or all $0$ for all cluster members $\{i: G_{i g}=1\}$. Because of this, the variances in Equation~\eqref{eq:var} cannot be estimated using standard sample variance estimator.

It often happens that the same problem presents itself to different communities and the solutions adopted by different communities are different. In this case, the econometrics literature promotes the cluster-robust variance estimator \citep{athey2017econometrics}, which can be dated back to the study of linear regression with clustered residuals \citep{liang1986longitudinal}. This cluster-robust variance estimator has the form of $A^{-1}\Sigma A^{-1}$ where $A$ and $\Sigma$ are both matrices, hence the other name sandwich estimator. 

Different sandwich estimators have been used in econometrics for decades, dated back to the famous Ericker-Huber-White estimator to tackle heteroskedasticity in linear regressions \citep{eicker1967limit,huber1967behavior,white1980heteroskedasticity,white1982maximum,kauermann2001note}.

\begin{theorem}[Cluster-Robust Variance Estimator]\label{thm:sandwich}
Let 
$$
\hat{\alpha} = \frac{\sum_{W_i = 0} Y_i}{N_C},
\quad
\hat{\tau} = \Delta = \frac{\sum_{W_i = 1} Y_i}{N_T}-\frac{\sum_{W_i = 0} Y_i}{N_C}.
$$
Let $\hat{\varepsilon}_{i}=Y_{i}^{\mathrm{obs}}-\hat{\alpha}-\hat{\tau} \cdot W_{i}$ be the residual.
The asymptotic covariance matrix of $\left(\hat{\alpha}, \hat{\tau} \right)$ can be estimated by:
\begin{equation}\label{eq:sandwich}
  \left(\sum_{i=1}^{N} \left( \begin{array}{cc}{1} & {W_{i}} \\ {W_{i}} & {W_{i}}\end{array}\right)\right)^{-1}\left(\sum_{g=1}^{G} \left[ \sum_{i: G_{i g}=1} \left( \begin{array}{c}{\hat{\varepsilon}_{i}} \\ {W_{i} \cdot \hat{\varepsilon}_{i}}\end{array}\right) \sum_{i: G_{ig}=1} \left( \begin{array}{c}{\hat{\varepsilon}_{i}} \\ {W_{i}^{i} \cdot \hat{\varepsilon}_{i}}\end{array}\right)^{\prime} \right ] \right)\left(\sum_{i=1}^{N} \left( \begin{array}{cc}{1} & {W_{i}} \\ {W_{i}} & {W_{i}}\end{array}\right)\right)^{-1}  
\end{equation}
In particular, the variance of $\Delta$ (same as $\hat{\tau}$ here) can be estimated using the $(2,2)$-th entry of \eqref{eq:sandwich}.
\end{theorem}

On the other camp, the A/B testing or online experimentation literature promotes the delta method \citep{kohavi2010online,Deng2017,Deng2018}, which tackles the variance estimation of the ATE estimator directly using large sample theory \citep{van2000asymptotic, dasgupta2008asymptotic}. The delta method is often presented as a way to extend asymptotic normality of a random variable (or vector of random variables) to a continuous function of the said variable(s). In the special case of clustered-randomized experiment, it has a simpler form by the following result.

\begin{lemma}\label{lem:dm}
For i.i.d. random variables $(S_k, D_k),k=1,\dots,n$, 
$$
\sqrt{n}
\left[
\frac{\sum S_k}{\sum D_k}\ 
- 
\frac{\sum
\left\{ 
S_k - \frac{\mathrm{E}(S)}{\mathrm{E}(D)} D_k 
\right\}}{n\mathrm{E}(D)} 
\right] 
\xrightarrow{P} 0\ .
$$
In other words, the asymptotic variance of $\frac{\sum S_k}{\sum D_k}$ is the same as that of 
$
\frac{\sum\left\{ S_k - \frac{\mathrm{E}(S)}{\mathrm{E}(D)} D_k \right\}}{n\mathrm{E}(D)}
$.
\end{lemma}

\begin{theorem}[Delta Method for Clustered-Randomized Experiments]\label{thm:dm}
Let $\hat{\alpha} = \frac{\sum_{W_i = 0} Y_i}{N_C}$, $\hat{\tau} = \Delta = \frac{\sum_{W_i = 1} Y_i}{N_T}-\frac{\sum_{W_i = 0} Y_i}{N_C}$. Let $\hat{\varepsilon}_{i}=Y_{i}^{\mathrm{obs}}-\hat{\alpha}-\hat{\tau} \cdot W_{i}$ be the residual.
Let $S_{gT} = \sum_{i: G_{ig}=1,W_i=1} \hat{\varepsilon}_{i}$ and $S_{gC} = \sum_{i: G_{ig}=1,W_i=0} \hat{\varepsilon}_{i}$.

$\mathrm{Var}\left (\frac{\sum_{W_i = 1} Y_i}{N_T}\right )$ and $\mathrm{Var}\left (\frac{\sum_{W_i = 0} Y_i}{N_C}\right )$ can be estimated by $\frac{\sum_g S_{gT}^2}{N_T^2}$ and $\frac{\sum_g S_{gC}^2}{N_C^2}$ respectively. $\mathrm{Var}(\Delta)$ can be estimated by 
\begin{equation}
\label{eqn:dm}
\frac{\sum_g S_{gT}^2}{N_T^2} + \frac{\sum_g S_{gC}^2}{N_C^2}
\end{equation}
\end{theorem}

It is straightforward to show that \eqref{eqn:dm}, after expansion, is the same as the delta method variance estimator using first order Taylor expansion. Theorem~\ref{thm:dm} merely offers a different way to derive the same estimator using Lemma~\ref{lem:dm}. We use this form because it makes the comparison to the Cluster-Robust variance estimator easier as both depends on the residual $\hat{\epsilon_i}$. 

The main result of this paper is the following equivalence result. 
\begin{theorem}[Equivalence of the two methods]\label{thm:main}
The $(2,2)$-th entry of \eqref{eq:sandwich}'s equals \eqref{eqn:dm}. 
\end{theorem}

\section{Proof of the main result}
We first prove Theorem~\ref{thm:main} using straightforward algebra. Then we prove Theorem~\ref{thm:dm} from Lemma~\ref{lem:dm}. The proof of Lemma~\ref{lem:dm} itself can be found in \citet[Appendix]{Deng2017}.

\begin{proof}[Proof of the main result Theorem~\ref{thm:main}]
\ \\
Starting from Equation~\eqref{eq:sandwich}
\begin{equation*}
  \left(\sum_{i=1}^{N} \left( \begin{array}{cc}{1} & {W_{i}} \\ {W_{i}} & {W_{i}}\end{array}\right)\right)^{-1}\left(\sum_{g=1}^{G} \left[ \sum_{i: G_{i g}=1} \left( \begin{array}{c}{\hat{\varepsilon}_{i}} \\ {W_{i} \cdot \hat{\varepsilon}_{i}}\end{array}\right) \sum_{i: G_{ig}=1} \left( \begin{array}{c}{\hat{\varepsilon}_{i}} \\ {W_{i}^{i} \cdot \hat{\varepsilon}_{i}}\end{array}\right)^{\prime} \right ] \right)\left(\sum_{i=1}^{N} \left( \begin{array}{cc}{1} & {W_{i}} \\ {W_{i}} & {W_{i}}\end{array}\right)\right)^{-1}  \ ,
\end{equation*} 
because $N_T = \sum W_i$ and $N = \sum_i 1$, 
$$
\sum_{i=1}^{N} \left( \begin{array}{cc}{1} & {W_{i}} \\ {W_{i}} & {W_{i}}\end{array}\right)  =  \left( \begin{array}{cc}{N_T+N_C} & {N_T} \\ {N_T} & {N_T}\end{array}\right)\,
$$
and 
$$
\left(\sum_{i=1}^{N} \left( \begin{array}{cc}{1} & {W_{i}} \\ {W_{i}} & {W_{i}}\end{array}\right) \right) ^{-1} = \left( \begin{array}{cc}{\frac{1}{N_C}} & {-\frac{1}{N_C}} \\ {-\frac{1}{N_C}} & {\frac{1}{N_T}+\frac{1}{N_C}}\end{array}\right) \ .
$$
For the middle term, we follow previous notations and let 
$$
S_{gT} = \sum_{i: G_{ig}=1,W_i=1} \hat{\varepsilon}_{i},
\quad
S_{gC} = \sum_{i: G_{ig}=1,W_i=0} \hat{\varepsilon}_{i}.
$$
For any cluster $g,$ all units $\{i: G_{ig}=1\}$ are either all assigned to treatment or control, which implies that either $S_{gT}=0$ or $S_{gC}=0.$ Consequently $S_{gT}S_{gC}=0$ for all $g=1, \ldots, G,$ and therefore
\begin{align*}
& \sum_{g=1}^{G} \left[ \sum_{i: G_{i g}=1} \left( \begin{array}{c}{\hat{\varepsilon}_{i}} \\ {W_{i} \cdot \hat{\varepsilon}_{i}}\end{array}\right) \sum_{i: G_{ig}=1} \left( \begin{array}{c}{\hat{\varepsilon}_{i}} \\ {W_{i}^{i} \cdot \hat{\varepsilon}_{i}}\end{array}\right)^{\prime} \right ] = \sum_{g=1}^{G} \left[ \left( \begin{array}{c} S_{gT}+S_{gC} \\ S_{gT}\end{array}\right) \cdot \left( \begin{array}{c} S_{gT}+S_{gC} \\ S_{gT}\end{array}\right)^{\prime} \right] \\
& = \sum_{g=1}^{G} \left( \begin{array}{cc}(S_{gT}+S_{gC})^2 & S_{gT}^2+S_{gT}S_{gC} \\ S_{gT}^2+S_{gT}S_{gC} & S_{gT}^2\end{array} \right) = \sum_{g=1}^{G} \left( \begin{array}{cc} S_{gT}^2+S_{gC}^2 & S_{gT}^2 \\ S_{gT}^2 & S_{gT}^2\end{array} \right) \ .
\end{align*}

The $(2,2)$-th entry of Equation~\eqref{eq:sandwich} is
\begin{align*}
\left( \begin{array}{cc} {-\frac{1}{N_C}} & {\frac{1}{N_T}+\frac{1}{N_C}}\end{array}\right)  \left( \begin{array}{cc}\sum_{g=1}^{G}( S_{gT}^2+S_{gC}^2) & \sum_{g=1}^{G} S_{gT}^2 \\ \sum_{g=1}^{G} S_{gT}^2 & \sum_{g=1}^{G} S_{gT}^2\end{array} \right)\left( \begin{array}{c} {-\frac{1}{N_C}} \\ {\frac{1}{N_T}+\frac{1}{N_C}}\end{array}\right) 
= \frac{\sum_g S_{gT}^2}{N_T^2} + \frac{\sum_g S_{gC}^2}{N_C^2} \ .
\end{align*}
This proves the mathematical form of the Cluster-Robust Variance Estimator is the same as the result of the Delta Method. 
\end{proof}

\begin{proof}[Proof of the Centered-form of the Delta Method variance estimator]
Let $W_g$ be the assignment indicator for cluster $g=1, \ldots, G,$ and 
$
G_T = \sum_{g=1}^G W_g
$ 
be the total number of clusters assigned to treatment. Note that if unit $i$ belongs to cluster $g$ where $W_g=1$, then by definition $W_i=1.$ Furthermore, we let
$
R_g = \sum_{i: G_{ig}=1, W_i = 1} Y_i,
$
and reformulate the treatment average as:
\begin{align}
\label{eq:use-delta-method}
\frac{\sum_{W_i = 1} Y_i}{N_T} 
&= \frac{\sum_{g:W_g=1} [\sum_{i: G_{ig}=1} Y_i]}{\sum_{g:W_g=1} [\sum_{g: G_{ig}=1} W_i] } \nonumber \\
&= \frac{\sum_{g:W_g=1} R_g }{\sum_{g:W_g=1} N_g}
\end{align}
Using Lemma~\ref{lem:dm}, the asymptotic variance of the above is equivalent to the asymptotic variance of 
\begin{align*}
\frac{\sum_{g:W_g=1} [R_g - \frac{\mathrm{E}_T(R_g)}{\mathrm{E}_T(N_g)} N_g]}{G_T\mathrm{E}_T(N_g)}
\end{align*}
where $\mathrm{E}_T$ denotes the expectation takes over treatment clusters. To be more specific, to reach this conclusion we note that the cluster-level aggregates $R_g$ and $N_g$ are i.i.d., and by first-order Taylor expansion
\begin{align*}
\frac{\sum_{g:W_g=1} R_g }{\sum_{g:W_g=1} N_g}
&= \frac{\bar R}{\bar N} \\
&\approx \frac{1}{\mathrm{E}_T N_g} (\bar R - \mathrm{E}_T R_g) - \frac{\mathrm{E}_T R_g}{\mathrm{E}_T^2 N_g} (\bar N - \mathrm{E}_T N_g) \\
&= \frac{1}{\mathrm{E}_T N_g } 
\left[
\frac{1}{G_T} \sum_{g: W_g=1} 
\left\{
R_g - \frac{\mathrm{E}_T R_g}{\mathrm{E}_T N_g} N_g
\right\}
\right]
\end{align*}
In other words, we have ``linearized'' the ratio in \eqref{eq:use-delta-method} by the Delta method. Consequently, 
\begin{equation}
\label{eq:use-delta-method-2}
\mathrm{Var}
\left(
\bar R  / \bar N
\right)
\approx
\frac{1}{G_T E_T^2 N_g}
\mathrm{Var}
\left(
R_g - \frac{\mathrm{E}_T R_g}{\mathrm{E}_T N_g} N_g
\right).
\end{equation}
To estimate the variance term in the right hand side of \eqref{eq:use-delta-method-2}, we adopt the plug-in approach. To be specific, we estimate $\mathrm{E}_T{N_g}$ by its empirical analogue 
$
\sum_{g:W_g=1}{N_g}/G_T,
$
and similarly
$
\frac{\mathrm{E}_T(R_g)}{\mathrm{E}_T(N_g)}
$
by
$
\frac{\sum_{W_i = 1} Y_i}{N_T} = \frac{ \sum_{g:W_g=1}{R_g}}{\sum_{g:W_g=1}{N_g}}.
$
It is worth noting that, for all $g$ such that $W_g = 1,$
\begin{align*}
R_g - \widehat{\frac{\mathrm{E}_T(R_g)}{\mathrm{E}_T(N_g)}} N_g 
&= \sum_{i: G_{ig}=1} 
\left\{
Y_i - \widehat{\frac{\mathrm{E}_T(R_g)}{\mathrm{E}_T(N_g)}} W_i 
\right\} \\
&= \sum_{i: G_{ig}=1} 
\left( 
Y_i - \frac{\sum_{W_i = 1} Y_i}{N_T} W_i 
\right) \\
&= \sum_{i: G_{ig}=1} 
\left( 
Y_i - \hat \alpha - \hat \tau W_i 
\right) \\
&= \sum_{i: G_{ig}=1} \hat{\epsilon_i}
\end{align*}
The second to last step holds because $W_i=1$ for all $i$ in cluster $g.$ Therefore, 
\begin{align}
\label{eq:use-delta-method-3}
\widehat{\mathrm{Var}}
\left(
R_g - \frac{\mathrm{E}_T R_g}{\mathrm{E}_T N_g} N_g
\right)
&=
\frac{1}{G_T} \sum_{g:W_g=1} 
\left\{
R_g - \widehat{\frac{\mathrm{E}_T(R_g)}{\mathrm{E}_T(N_g)}} N_g
\right\}^2 \nonumber \\
&= \frac{1}{G_T} \sum_{g:W_g=1} 
\left(
\sum_{i: G_{ig}=1} \hat{\epsilon_i}
\right)^2.
\end{align}
By \eqref{eq:use-delta-method-2} and \eqref{eq:use-delta-method-3}, 
\begin{align*}
\widehat{\mathrm{Var}}
\left(
\frac{\sum_{W_i = 1} Y_i}{N_T} 
\right)
&\approx \frac{1}{G_T^2 \hat{\mathrm{E}}_T^2 R_g} \sum_{g:W_g=1} 
\left(
\sum_{i: G_{ig}=1} \hat{\epsilon_i}
\right)^2 \\
&= \frac{\sum_{g:W_g=1} \left (\sum_{i: G_{ig}=1} \hat{\epsilon_i}\right)^2}{(\sum_{g:W_g=1}{N_g})^2} \\
&= \frac{\sum_g S_{gT}^2} {N_T^2}.
\end{align*}

Similarly, the delta method variance estimation (using plug-in estimator) of $\frac{\sum_{W_i = 0} Y_i}{N_C}$ is $\frac{\sum_g S_{gC}^2} {N_C^2}$. Therefore the delta method variance estimation of $\Delta$ is the sum of the two as in Equation~\eqref{eqn:dm}.
\end{proof}

\appendix
\section{R Code and Example Results}

\begin{lstlisting}[language=R]
library(dplyr)
set.seed(1)
numCluster = 100
cluster_size = rpois(numCluster,10)
ctr = runif(numCluster,0.5,1)
N = sum(cluster_size)
ctr_user = rep(ctr,cluster_size)

clustervec = rep(1:numCluster,cluster_size)
assignment = rep(0:1, length.out=numCluster)
assignVec = rep(assignment, cluster_size)
yvec = rbinom(N, 1, ctr_user)

simdata = tibble(Cluster = clustervec, Y = yvec, W = assignVec)



#' w is binary assignment, y is response, c is cluster id
sandwich = function(y,w,cl){
  eps = lm(y~w)$resid
  sumw = sum(w)
  A = matrix(c(length(w),sumw, sumw,sumw),2)
  Ainv = solve(A)
  .data = tibble(y=y, w=w, cl=cl,eps=eps)
  B = .data %>% group_by(cl) %>% 
    summarise(B11=sum(eps)^2,B12=sum(eps*w)*sum(eps),B21 = B12, B22 = sum(w*eps)^2) %>%
    summarise_at(vars(-cl),sum) %>% as.matrix() %>% matrix(nrow=2) 
  Ainv %*% B %*% Ainv
}

#' Delta Method
DM = function(y,w,cl, pop=TRUE){
  .data = tibble(y=y, w=w, cl=cl)
  treat = .data %>% filter(w==1)
  control = .data %>% filter(w==0)
  DMinner = function(y,cl){
    clevel = tibble(y,cl) %>% group_by(cl) %>% summarise(ys = sum(y),ws = n()) 
    #wsum is the count of users in each cluster
    muy = mean(clevel$ys)
    muw = mean(clevel$ws)
    vary = var(clevel$ys)
    varw = var(clevel$ws)
    covyw = cov(clevel$ws, clevel$ys)
    n = length(clevel$cl)
    if(pop){
      vary = vary * (n-1)/n
      varw = varw * (n-1)/n
      covyw = covyw * (n-1)/n
    }
    vary/n/muw^2-2*muy/muw^3*covyw/n+muy^2/muw^4*varw/n
  }
  DMinner(treat$y, treat$cl) + DMinner(control$y, control$cl)
}

#' Sandwich estimator implemented using the simplified form without matrix algebra
sandwich2 = function(y,w,cl){
  eps = lm(y~w)$resid
  .data = tibble(eps=eps, w=w, cl=cl)
  tmp = .data %>% group_by(cl) %>% summarise(sgt = sum(eps*w),sgc = sum(eps*(1-w)))
  sum(tmp$sgt^2)/sum(w)^2 + sum(tmp$sgc^2)/sum(1-w)^2
}


> sandwich(simdata$Y, simdata$W,simdata$Cluster)[2,2]
[1] 0.001419918
> sandwich2(simdata$Y, simdata$W,simdata$Cluster)
[1] 0.001419918
> DM(simdata$Y, simdata$W,simdata$Cluster)
[1] 0.001419918
> DM(simdata$Y, simdata$W,simdata$Cluster,pop=FALSE)
[1] 0.001448896

\end{lstlisting}

\bibliographystyle{plainnat}  
\bibliography{references}  %%% Remove comment to use the external .bib file (using bibtex).

\begin{thebibliography}{13}
\providecommand{\natexlab}[1]{#1}
\providecommand{\url}[1]{\texttt{#1}}
\expandafter\ifx\csname urlstyle\endcsname\relax
  \providecommand{\doi}[1]{doi: #1}\else
  \providecommand{\doi}{doi: \begingroup \urlstyle{rm}\Url}\fi

\bibitem[Athey and Imbens(2017)]{athey2017econometrics}
Susan Athey and Guido~W Imbens.
\newblock The econometrics of randomized experiments.
\newblock In \emph{Handbook of Economic Field Experiments}, volume~1, pages
  73--140. 2017.

\bibitem[Dasgupta(2008)]{dasgupta2008asymptotic}
Anirban Dasgupta.
\newblock \emph{Asymptotic Theory of Statistics and Probability}.
\newblock Springer Science \& Business Media, 2008.

\bibitem[Deng et~al.(2017)Deng, Lu, and Litz]{Deng2017}
Alex Deng, Jiannan Lu, and Jonthan Litz.
\newblock Trustworthy analysis of online {A}/{B} tests: Pitfalls, challenges
  and solutions.
\newblock In \emph{Proceedings of the Tenth ACM International Conference on Web
  Search and Data Mining}, pages 641--649, 2017.

\bibitem[Deng et~al.(2018)Deng, Knoblich, and Lu]{Deng2018}
Alex Deng, Ulf Knoblich, and Jiannan Lu.
\newblock Applying the delta method in metric analytics: A practical guide with
  novel ideas.
\newblock In \emph{Proceedings of the 24th ACM SIGKDD International Conference
  on Knowledge Discovery and Data Mining}, pages 233--242, 2018.

\bibitem[Eicker(1967)]{eicker1967limit}
Friedhelm Eicker.
\newblock Limit theorems for regressions with unequal and dependent errors.
\newblock In \emph{Proceedings of the Fifth Berkeley Symposium on Mathematical
  Statistics and Probability}, volume~1, pages 59--82, 1967.

\bibitem[Huber(1967)]{huber1967behavior}
Peter~J Huber.
\newblock The behavior of maximum likelihood estimates under nonstandard
  conditions.
\newblock In \emph{Proceedings of the Fifth Berkeley Symposium on Mathematical
  Statistics and Probability}, volume~1, pages 221--233. University of
  California Press, 1967.

\bibitem[Kauermann and Carroll(2001)]{kauermann2001note}
G{\"o}ran Kauermann and Raymond~J Carroll.
\newblock A note on the efficiency of sandwich covariance matrix estimation.
\newblock \emph{Journal of the American Statistical Association}, 96:\penalty0
  1387--1396, 2001.

\bibitem[Klar and Donner(2001)]{klar2001current}
Neil Klar and Allan Donner.
\newblock Current and future challenges in the design and analysis of cluster
  randomization trials.
\newblock \emph{Statistics in Medicine}, 20:\penalty0 3729--3740, 2001.

\bibitem[Kohavi et~al.(2010)Kohavi, Longbotham, and Walker]{kohavi2010online}
Ron Kohavi, Roger Longbotham, and Toby Walker.
\newblock Online experiments: Practical lessons.
\newblock \emph{Computer}, 43\penalty0 (9):\penalty0 82--85, 2010.

\bibitem[Liang and Zeger(1986)]{liang1986longitudinal}
Kung-Yee Liang and Scott~L Zeger.
\newblock Longitudinal data analysis using generalized linear models.
\newblock \emph{Biometrika}, 73\penalty0 (1):\penalty0 13--22, 1986.

\bibitem[Van~der Vaart(2000)]{van2000asymptotic}
Aad~W Van~der Vaart.
\newblock \emph{Asymptotic Statistics}.
\newblock Cambridge University Press, 3rd edition, 2000.

\bibitem[White(1980)]{white1980heteroskedasticity}
Halbert White.
\newblock A heteroskedasticity-consistent covariance matrix estimator and a
  direct test for heteroskedasticity.
\newblock \emph{Econometrica}, 48\penalty0 (4):\penalty0 817--838, 1980.

\bibitem[White(1982)]{white1982maximum}
Halbert White.
\newblock Maximum likelihood estimation of misspecified models.
\newblock \emph{Econometrica}, 50\penalty0 (1):\penalty0 1--25, 1982.

\end{thebibliography}


\end{document}